\documentclass[11pt]{article}
\usepackage[utf8]{inputenc}
\usepackage[T1]{fontenc}
\usepackage{latexsym}
\usepackage{amssymb} 
\usepackage{mathrsfs}
\usepackage{amsmath}
\usepackage{latexsym}
\usepackage{delarray}
\usepackage{amssymb,amsmath,amsfonts,amsthm,mathrsfs}
\usepackage{scalerel}
\usepackage{tabularx}
\usepackage[table]{xcolor}
\usepackage{float} 
\usepackage{overpic} 
\usepackage{caption} 
\usepackage{graphicx}
\usepackage{pgfplots}
\pgfplotsset{compat=1.17} 

\usepackage{tikz}
\usetikzlibrary{shapes.geometric, arrows.meta, positioning}

\usepackage[font=small,labelfont=bf]{caption}

\newcommand\reallywidehat[1]{\arraycolsep=0pt\relax%
\begin{array}{c}
\stretchto{
  \scaleto{
    \scalerel*[\widthof{\ensuremath{#1}}]{\kern-.5pt\bigwedge\kern-.5pt}
    {\rule[-\textheight/2]{1ex}{\textheight}} 
  }{\textheight} %
}{0.5ex}\\           
#1\\                 
\rule{-1ex}{0ex}
\end{array}
}

\usepackage{hyperref}
\renewcommand\eqref[1]{(\ref{#1})} 

 \newtheorem{thm}{Theorem}[section]
 \newtheorem{cor}[thm]{Corollary}
 \newtheorem{lem}[thm]{Lemma}
 \newtheorem{prop}[thm]{Proposition}
 \theoremstyle{definition}
 \newtheorem{defn}[thm]{Definition}
 \theoremstyle{remark}
 
 \newtheorem{ex}[thm]{Example}
 \numberwithin{equation}{section}

\newcommand{\bi}{\begin{itemize}}
\newcommand{\ei}{\end{itemize}}
\newcommand{\be}{\begin{enumerate}}
\newcommand{\ee}{\end{enumerate}}
\newcommand{\beq}{\begin{equation}}
\newcommand{\eq}{\end{equation}}

\def\SU2{{{\rm SU(2)}}}
\def\SO3{{{\rm SO(3)}}}
\def\lapsu2{{{\mathcal L}_\SU2}}

\begin{document}
\title{On The Fast Fourier Transform on SU(2)}

\author{Julio Delgado\thanks{Departamento de Matem\'aticas, Universidad del Valle; \texttt{delgado.julio@correounivalle.edu.co}} \and Alejandro Uma\~na\thanks{Departamento de Matem\'aticas, Universidad del Valle; \texttt{alejandro.umana@correounivalle.edu.co}}}

\date{\today}
\maketitle
\begin{abstract}
The special unitary group SU(2) plays a fundamental role in the description of symmetries in quantum mechanics, theoretical physics, and spherical signal processing. In this paper, we address the computational challenges of performing spectral analysis on this non-abelian compact Lie group. We present the Fourier Transform (FT) on SU(2) and develop a Fast Fourier Transform (FFT) algorithm inspired by the classical Cooley-Tukey divide-and-conquer scheme. Our approach efficiently discretizes the group using Euler angles, applying a two-dimensional FFT on the angular variables and exploiting the recursive properties of Jacobi polynomials. We provide an analysis of the computational complexity, demonstrating that our FFT-based method significantly outperforms the direct computation of the FT. This algorithm serves as a foundational tool for understanding the implementation of the FFT on SU(2), a key component in numerical simulations and advanced data analysis for high-performance computing applications on curved manifolds and quantum systems.
\end{abstract}


\bigskip
\section{Introduction}

In this work we develop an algorithm that follows the classical Cooley-Tukey divide-and-conquer scheme of the standard Fast Fourier Transform (FFT). Our construction is inspired by fast Fourier transform techniques on the 2-Sphere  and by the Quantum Fourier Transform (QFT) on SU(2), from the developments in a purely classical setting in the Theorem \ref{thm:final_legendre_fft} (cf. \cite{2-sphere},\cite{Quantum_SU(2)}). 
We also analyze and compare the computational complexity of the direct FT and the proposed FFT-based method on SU(2). 
Crucially, we show that by employing Propositions \ref{P_1} and \ref{P_2}, we can prove Theorem \ref{TEO_FIN}. This theorem establishes a method to reduce the computational complexity of the direct transform from $\mathcal{O}(N^6)$ to $\mathcal{O}(N^4)$.\\

The Quantum Fourier Transform is the quantum analogue of the  discrete Fourier transform and due to the recent intensive research on quantum computation, the interest in Fast Fourier transforms on special groups arises (cf. \cite{ft-QFT},  \cite{Bastidasquantum}). Recently, it
has been shown that QSP sequences for su(2) and su(1,1) are
intimately related to the nonlinear Fourier (\cite[p. 1]{bastidas2024complexification}).\\

The Fast Fourier Transform (FFT) was commemorated with an IEEE Milestone during a ceremony held in May 2025 at Princeton University (cf. \cite{FFTmile}). This computer algorithm is found in just about every electronic device and it  has become an important tool for manipulating and analyzing signals in many areas including audio processing, telecommunications, digital broadcasting, and image analysis. By exploiting  algebraic properties and periodicities, the FFT reduced the number of the operations, making it particularly and practically feasible for everyday tasks, replacing the less efficient analog methods (cf.  \cite{Cooley-Tukey}, \cite{FFT_Aplications}).\\%

The history of the Fourier Transform (FT) dates back to the early 19th century, when the French mathematician Jean-Baptiste Joseph Fourier, during Napoleon Bonaparte's military campaign in Egypt (1798--1801), began studying the propagation of heat in solid media (cf. \cite{History_Fourier}). During this expedition, Fourier was appointed secretary of the \textit{Institut d'Égypte}, where he combined administrative and diplomatic duties with intense scientific activity. It was in this context that he developed the revolutionary idea of representing general periodic functions as sums of simple waves, thus laying the groundwork for what would later become known as the Fourier Series (cf. \cite{History_Fourier,FT_History_Importance}).\\%

After returning to France, Fourier consolidated these concepts and, in 1807, presented his seminal work on the analytical theory of heat, formally establishing the principle that we now call the Fourier Transform (FT) (cf. \cite{History_Fourier, FT_History_Importance}). With the advancement of technology and the increasing availability of data throughout the 20th century, the FT gained enormous practical importance, particularly during the Cold War. In the 1960s, one of its most strategic applications arose in the context of detecting covert nuclear tests (cf. \cite{History_Cold_War,Cooley-Tukey}). Underground detonations generated seismic waves that propagated through the Earth's crust; these signals were recorded by sensors and analyzed using the Discrete Fourier Transform (DFT) to detect frequency patterns characteristic of nuclear explosions (cf. \cite{Nuclear_Detection_FFT}). However, the direct computation of the DFT has a computational complexity of
\[
\mathcal{O}(N^2),
\]
which hindered its efficient use on a large scale (cf. \cite{oppenheimN^2, whatistheFFT,Cooley-Tukey}).\\

This need for efficiency led James Cooley and John Tukey in 1965 to develop a faster algorithm: the Fast Fourier Transform (FFT), which reduces the number of operations to
\[
\boldsymbol{\mathcal{O}(N \log_2 N),}
\]
by exploiting symmetries and recursive structures in the data (cf. \cite{Cooley-Tukey,FFThistorical}). Interestingly, Carl Friedrich Gauss had already outlined a similar method nearly a century earlier, although his work was not recognized as a precursor to the FFT until much later (cf. \cite{GAUSS}).\\%


Today, the FFT remains a widely used algorithm. Its efficiency in computing the transform of a time series makes it a fundamental technique in modern technology, enabling critical applications such as faster data transmission (cf. \cite{whatistheFFT}).\\%

To illustrate the significant difference between the Fourier Transform (FT) and the Fast Fourier Transform (FFT), let us consider a signal composed of $N = 2^{30}$ samples. Assuming that each operation takes 1 nanosecond, the total time required to compute the FT would be approximately 13,343 days. In contrast, using the FFT drastically reduces the number of operations, bringing the total computation time down to approximately 64 seconds.
\begin{table}[H]
    \centering
     \caption{Comparison between the number of operations required by the FT and the FFT for various values of \( N \).}
    \begin{tabularx}{\textwidth}{|c|X|X|}
        \hline \rowcolor{red!15} 
        \textbf{\( N \)} & \textbf{ \textbf{FT Operations ($\boldsymbol{N^2}$)}} & \textbf{\textbf{FFT Operations ($\boldsymbol{N \log _2N}$)}} \\
        \hline
        \rowcolor{gray!10} \( 2^2 \) & 16 & 8  \\
        \( 2^4 \) & 256 & 64 \\
        \rowcolor{gray!10} \( 2^8 \) & 65536 & 2048  \\
        \( 2^{16} \) & 4294967296 & 1048576 \\
        \hline
    \end{tabularx}
\end{table}

The FFT is also an useful tool in engineering, physics, and computational sciences. Its applications range from signal and image processing, data compression, and audio synthesis to the numerical solution of partial differential equations and the simulation of dynamic systems (cf. \cite{FFT_Aplications,FFTphoto}).\\%

In the classical case, the typical domain for the Fourier Transform is the torus
$\mathbb{T}$, which, as a compact abelian group, has the property that all of its irreducible representations are one-dimensional, known as characters (cf. \cite{Toro1,Toro2,stein2011fourier}). This algebraic simplicity is crucial
which leads to the Discrete Fourier Transform (DFT) when the function is sampled
(cf. \cite{ruzhansky2009pseudo, Toro1}).\\%

However, a natural question arises: what happens when the group is non-abelian? In such cases, as with the group
SU(2), the irreducible representations are no longer one-dimensional, but instead matrices of size $(2l+1) \times (2l+1)$ for half-integer $l \in \tfrac{1}{2}\mathbb{N}_0$. 
The group SU(2) plays a central role in various areas of theoretical physics, particularly in quantum mechanics and quantum computing (cf. \cite{bastidas2024complexification, Bastidasquantum}). Its structure as a compact, non-abelian Lie group and the matrix nature of its irreducible representations make SU(2) an ideal case for extending the Fourier Transform to non-commutative settings (cf.\cite{Faraut_2008}).\\%

In this paper, we develop a Fast Fourier Transform (FFT) algorithm specifically adapted to the group SU(2). 
The proposed method follows the classical Cooley--Tukey divide-and-conquer scheme, suitably modified to account for the non-abelian group structure, and allows for a detailed analysis of the arithmetic operation count required for its execution. Beyond its theoretical relevance, the study of the Fourier Transform on SU(2) provides a foundation for the design of fast algorithms in non-commutative settings, with potential applications in quantum signal processing, physical simulations, and quantum information theory (cf.\cite{Cite_Aplications_inverse, Quantum_SU(2)}).\\

We now give an outline of the organization of the paper. In Section \ref{SEC:Prelim}, we review the basic elements of Fourier analysis on the torus to provide the necessary historical and conceptual background \cite{Toro1,Toro2, Cooley-Tukey}. In Section \ref{SEC:FAus2}, we specialize the analysis to SU(2), establishing the group-theoretic foundations, the Peter-Weyl theorem, and the specific discretization via Euler angles \cite{Faraut_2008,ruzhansky2009pseudo,Preli}. In Section \ref{sec:FFTSU2}, we present the core of contributions in this work: the construction of the Fast Fourier Transform (FFT) algorithm on SU(2). This section details the divide-and-conquer scheme and the recursive relations of Jacobi polynomials that enable the acceleration \cite{2-sphere,Quantum_SU(2),FFT_Compact_Groups}. Section \ref{sec:order-FT} is devoted to a rigorous computational complexity analysis of the direct transform FT, where we establish its asymptotic order. Section \ref{sec:order-fftsu2} extends this analysis to our proposed algorithm, proving the reduction in operation count and contrasting it with the direct method. Finally, in Section \ref{sec:ftvsfft}, we illustrate the computational gap between the standard FT and the proposed FFT.


\section{Preliminaries}  
\label{SEC:Prelim}
In this section, we start by recalling some basics on the Fourier transform on the torus \( \mathbb{T} \) and the compact Lie group SU(2), both through the lenses of representation theory. In each case, we employ the Peter–Weyl theorem to describe the corresponding Fourier expansion in terms of irreducible unitary representations. For further background, we refer the reader to (cf. \cite{FFT_Compact_Groups}).
\medskip

Let us begin by introducing the basic framework of representation theory for compact Lie groups, which will serve as a foundation for understanding the Fourier transform both on the torus and on the group SU(2). Our approach relies on unitary representations and culminates in the Peter–Weyl Theorem, which enables the decomposition of square-integrable functions on a compact group in terms of its irreducible representations.\\

Let $G$ be a compact group, and let $\mathcal{V}$ be a normed vector space over $\mathbb{R}$ or $\mathbb{C}$. We equip $G$ with its normalized Haar measure, denoted by $\mu \equiv dx$. The linear group of all invertible linear operators on $\mathcal{V}$ is denoted by $GL(\mathcal{V})$.

\begin{defn}
    A \textit{representation} of $G$ on $\mathcal{V}$ is a group homomorphism
\[
\pi : G \longrightarrow GL(\mathcal{V}),
\]
such that the map $g \mapsto \pi(g)v$ is continuous for every $v \in \mathcal{V}$. 
\end{defn}
The pair $(\pi, \mathcal{V}_\pi)$ denotes the representation, where $\mathcal{V}_\pi$ is referred to as the \textit{representation space}, and $d_\pi = \dim(\mathcal{V}_\pi)$ is called the \textit{dimension} or \textit{degree} of the representation. 
Let $\{e_1, e_2, \dots, e_{d_\pi}\}$ be an orthonormal basis of $\mathcal{V}_\pi$. Then the operator $\pi(g)$ can be written as a matrix with respect to this basis
\[
\pi(g) = \begin{pmatrix}
\pi_{1,1}(g) & \pi_{1,2}(g) & \dots & \pi_{1,d_\pi}(g) \\
\pi_{2,1}(g) & \pi_{2,2}(g) & \dots & \pi_{2,d_\pi}(g) \\
\vdots & \vdots & \ddots & \vdots \\
\pi_{d_\pi,1}(g) & \pi_{d_\pi,2}(g) & \dots & \pi_{d_\pi,d_\pi}(g)
\end{pmatrix},
\]
where the entries $\pi_{i,j}(g) := \langle \pi(g) e_j, e_i \rangle_{\ell^2(G)}$ are called the \textit{matrix coefficients} of the representation.  For fixed indices $i,j$, each $\pi_{i,j}$ defines a complex-valued function on $G$, and one can compute the inner product of two such functions as
\[
\langle \pi_{i,j}, \pi_{k,l} \rangle_{\ell^2(G)} = \int_G \pi_{i,j}(g) \overline{\pi_{k,l}(g)}\, d\mu(g) = \frac{1}{d_\pi} \delta_{j,l} \delta_{i,k},
\]
where $\delta$ is the Kronecker delta. This orthogonality relation implies that the set $\{ \pi_{i,j} \}$ forms an orthogonal system in $\ell^2(G)$ (cf. \cite{Preli}).

\begin{defn}
    Let $M_\pi$ denote the subspace of $\ell^2(G)$ spanned by the matrix coefficients $\{\pi_{i,j} : 1 \leq i,j \leq d_\pi\}$.
\end{defn}
Since these functions are pairwise orthogonal and satisfy $\| \pi_{i,j} \|^2 = \frac{1}{d_\pi}$, the normalized set
\[
\left\{ \sqrt{d_\pi}\, \pi_{i,j} : 1 \leq i,j \leq d_\pi \right\}
\]
constitutes an orthonormal basis for $M_\pi$ (cf. \cite{Preli}).

We denote by $\widehat{G}$ the set of equivalence classes of irreducible unitary representations of $G$. Since $G$ is compact, the set $\widehat{G}$ is discrete  (cf. \cite{DualDiscreto1939}). The Peter–Weyl theorem provides a complete characterization of $\ell^2(G)$ in terms of these irreducible representations (cf. \cite{Preli,Faraut_2008}).

\begin{thm}[Peter–Weyl]\label{PETER}
Let $G$ be a compact group. Then
\[
\ell^2(G) = \widehat{\bigoplus}_{\pi \in \widehat{G}} M_\pi,
\]
where the sum is orthogonal and complete in $\ell^2(G)$.
\end{thm}
(see, e.g., \cite{Preli,Faraut_2008}).\\

This result not only yields a decomposition of $\ell^2(G)$ into finite-dimensional subspaces but also provides the foundation for Fourier analysis on compact groups. Each irreducible representation contributes a complete set of orthogonal functions to the space $\ell^2(G)$, and the collection of all such matrix coefficients (properly normalized) forms a global orthonormal  (cf. \cite{peterweylOriginal}).\\

Having introduced the main tools that will be used throughout this paper, we begin by defining the $n$-\textit{dimensional Torus}, which will allow us to recall the Fourier transform in a more classical setting.

\begin{defn}
    The $n$-dimensional Torus, denoted by $\mathbb{T}^n$, is defined as the Cartesian product of $n$ copies of the unit circle
    \[
        \mathbb{T}^n = \underbrace{\mathbb{S}^1 \times \cdots \times \mathbb{S}^1}_{n \text{ times}}.
    \]
\end{defn}

Since $\mathbb{T}^n$ is an abelian compact group, its irreducible unitary representations are one-dimensional (cf. \cite{Faraut_2008,LieHall}). This leads to the following classification

\begin{prop}
    The irreducible representations of $\mathbb{T}^n$ are the continuous homomorphisms $\chi: \mathbb{T}^n \to \mathbb{S}^1$. They are indexed by $\mathbb{Z}^n$ and are defined by
    \[
        \chi_k(x) = e^{2\pi i\, k \cdot x} = \exp\left(2\pi i \sum_{j=1}^n k_j x_j\right),
    \]
    for any $x = (x_1, \dots, x_n) \in \mathbb{T}^n$ and $k \in \mathbb{Z}^n$.
\end{prop}
The collection $\{ \chi_k \}_{k \in \mathbb{Z}^n}$ forms an orthonormal basis of $\ell^2(\mathbb{T}^n)$ (cf. \cite{Toro1,Toro2}), and this is the classical Fourier basis on the Torus thus, by the Peter–Weyl theorem\ref{PETER}, the function $f \in \ell^2(\mathbb{T}^n)$ can be expressed in terms of the orthonormal basis (\textit{Fourier series}) as follows
\[
f(x) = \sum_{k \in \mathbb{Z}^n} \widehat{f}(k) e^{2\pi i \langle k, x \rangle},
\]
where
\begin{itemize}
    \item  $x \in \mathbb{T}^n \simeq [0,1)^n$,
\item $k \in \mathbb{Z}^n$ is a multi-index,
\item $\langle k, x \rangle = \sum_{j=1}^n k_j x_j$ is the standard Euclidean inner product,
\item and $\widehat{f}(k)$ denotes the \textit{Fourier coefficient} of $f$, given by
\[
\widehat{f}(k) = \int_{\mathbb{T}^n} f(x) e^{-2\pi i \langle k, x \rangle} \, dx.
\]
\end{itemize}
(see, e.g., \cite{ruzhansky2009pseudo}).\\%

In practice, the Fourier transform on the Torus $\mathbb{T}^n$ can be approximated numerically by sampling the function on a uniform grid. Let $f \in \ell^2(\mathbb{T}^n)$ be a square-integrable function, and consider a uniform discretization of the Torus with $N$ points per dimension. Define the grid
\[
x_j := \left( \frac{j_1}{N}, \frac{j_2}{N}, \dots, \frac{j_n}{N} \right), \quad j = (j_1, \dots, j_n) \in \{0, 1, \dots, N-1\}^n.
\]
We denote the discrete samples of $f$ as $f_j := f(x_j)$.

\begin{defn}
    The \textit{Discrete Fourier Transform} (DFT) of $f$ on the Torus is then given by
\[
\widehat{f}(k) := \sum_{j \in \{0, \dots, N-1\}^n} f_j \, e^{-2\pi i \langle k, x_j \rangle}, \quad k \in \{0, \dots, N-1\}^n.
\]
\end{defn}
(see, e.g., \cite{Toro1,Toro2,FFTmultidimensional}).\\%

Computing this expression directly requires evaluating $N^n$ sums, each involving $N^n$ complex multiplications and additions. Therefore, the total computational cost of a direct evaluation is of order $O(N^{2n})$.\\%

Numerically, the discrete version above can be efficiently computed using the \textit{Fast Fourier Transform} (FFT), such as the Cooley–Tukey algorithm (cf. \cite{Cooley-Tukey}). However, due to the separable structure of the DFT, it is possible to compute the transform more efficiently using the Fast Fourier Transform (FFT) algorithm in multiple dimensions (cf. \cite{Cooley-Tukey,Toro2}). This is achieved by applying the one-dimensional FFT successively along each axis of the grid.\\%

For each dimension $d = 1, \dots, n$, the following steps are performed
\begin{itemize}
    \item Fix the $n - 1$ other indices and apply a one-dimensional FFT of size $N$ along the $d$-th axis.
 \item There are $N^{n-1}$ such combinations of fixed indices, requiring $N^{n-1}$ one-dimensional FFTs.
\item  Each of these FFTs has a cost of $O(N \log N)$, so the total cost per dimension is
\[
O(N^{n-1} \cdot N \log N) = O(N^n \log N).
\]
\end{itemize}
Repeating this process for all $n$ dimensions yields a total computational complexity of
\[
O(n N^n \log N).
\]
Since in most applications the number of dimensions $n$ is fixed and relatively small, the complexity simply as
\[
\boldsymbol{O(N^n \log N).}
\]
For further background, we refer the reader to \cite{FFTmulti-general}.\\%

\section{Fourier Analysis on SU(2)}\label{SEC:FAus2}
In this section, we study the Fourier Analysis on SU(2), starting with the definition of its representations and subsequently formulating the corresponding Fourier Transform on SU(2). Further details can be found in \cite{ruzhansky2009pseudo}.
\medskip

Beyond the classical setting of the Torus, it is often of interest to generalize Fourier analysis to more sophisticated compact Lie groups. One of the most fundamental examples in this context is the \textit{special unitary group of degree two}, denoted by $\text{SU}(2)$. Formally, 

 \begin{defn}
The special unitary group $\text{SU}(2)$ is defined as the group of all $2 \times 2$ complex matrices $U$ satisfying
\[
U^\dagger U = I, \quad \text{and} \quad \det(U) = 1,
\]
where $U^\dagger$ denotes the conjugate transpose of $U$, and $I$ is the identity  (cf. \cite{ruzhansky2009pseudo,Faraut_2008}).
 \end{defn}
 
 These conditions imply that $\text{SU}(2)$ is a compact, connected, non-abelian Lie group of real dimension 3 (cf.\cite{Preli}).\\%

In this setting, the Peter–Weyl theorem guarantees that the space $\ell^2(\text{SU}(2))$ decomposes into an orthonormal basis formed by matrix coefficients of the irreducible unitary representations of the group, providing a non-commutative generalization of Fourier series (cf. \cite{ruzhansky2009pseudo}).\\%

We begin by identifying $z= (z_1,z_2) \in \mathbb{C}^2$ with the matrix $z=(z_1 \quad z_2 ) \in \mathbb{C}^{1\times 2}$, and we consider the following map
\begin{align*}
    T: \text{SU(2)}  \rightarrow GL(\mathbb{C}[z_1, z_2]), \quad (T(u)f)(z) &:= f(zu)\\
    &:=f(az_1 + cz_2, bz_1 + dz_2),
\end{align*}
where $u=\begin{pmatrix}
    a & b \\
    c &d
\end{pmatrix}
\in  \text{SU(2)} $ y \( GL(\mathbb{C}[z_1, z_2]) \).

 Note that $T$ is a representation of SU(2) on $\mathbb{C}[z_1,z_2]$, necessarily reducible as $\mathbb{C}[z_1,z_2]$ is infinite \cite[p. 612]{ruzhansky2009pseudo}.

For each $l \in \frac{1}{2}\mathbb{N}_0$, let $\mathcal{V}_l$ be the subspace of $\mathbb{C}[z_1,z_2]$ consisting of homogeneous polynomials of degree $2l \in \mathbb{N}_0$, that is
\[
V_l = \left\{ f \in \mathbb{C}[z_1, z_2] : f(z_1, z_2) = \sum_{k=0}^{2l} a_k z_1^k z_2^{2l-k}, \, \{a_k\}_{k=0}^{2l} \subset \mathbb{C} \right\}.
\]
(see, e.g., \cite{ruzhansky2009pseudo}).

We denoted by $T_l$ the restriction of the $T$ to the $T$-invariant subspace $\mathcal{V}_l$, which has dimension $2l+1 \in \mathbb{Z}^{+}$
\[
T_l: \text{SU(2)}  \rightarrow GL(V_l), \quad (T_l(u)f)(z) = f(zu),
\]
we now recall that $T_l$ is irreducible, unitary with respect to a natural inner product of $\mathcal{V}_l$, and that (up to unitary equivalence) there are no other irreducible unitary representations for SU(2) (\cite[p. 612]{ruzhansky2009pseudo}). By considering each $f\in \mathcal{V}_l$ as a function on SU(2), we endow $\mathcal{V}_l$ with the $\ell^2$-inner product. A natural basis for the vector space $\mathcal{V}_l$ is given by the monomial $\{p_{lk}:k\in0,1,...,2l\}$, where
\begin{align}\label{Polinomio_Plk}
    p_{lk}(z) = z_1^k z_2^{2l-k}.
\end{align} 
(see, e.g., \cite{ruzhansky2009pseudo}).\\

Now that we have constructed an orthonormal basis for the representation spaces $\mathcal{V}_l$, and having defined the unitary irreducible representations of SU(2), we are in a position to state the following theorem
\begin{thm}
Let \( u \in  \text{SU(2)}  \) defined by
\[
u = u(\phi, \theta, \psi) = \begin{pmatrix} 
a & b \\
c & d 
\end{pmatrix} = \begin{pmatrix} 
e^{i(\phi+\psi)/2} \cos \frac{\theta}{2} & e^{i(\phi-\psi)/2} i \sin \frac{\theta}{2} \\
e^{-i(\phi-\psi)/2} i \sin \frac{\theta}{2} & e^{-i(\phi+\psi)/2} \cos \frac{\theta}{2} 
\end{pmatrix}.
\]
Then, the \textbf{matrix coefficients} \( t^l_{mn}(u) \)IQ9 are defined as
\[
t^l_{mn}(u) = \frac{d^{l-m}}{dz_1^{l-m}} \frac{d^{l+m}}{dz_2^{l+m}}  \frac{(z_1a + z_2c)^{l-n} (z_1b + z_2d)^{l+n}}{\sqrt{(l - m)!(l + m)!(l - n)!(l + n)!}}.
\]
\end{thm}

This result can be further extended using Euler angles, leading to the following expression
\[t_{mn}^l(\phi, \theta, \psi)=P_{mn}^l(\cos(\theta))e^{-i(m\phi+n\psi)},\]
where
\[
P_{mn}^l(x)=c_{mn}^l\frac{(1-x)^{(n-m)/2}}{(1+x)^{(m+n)/2}}\left( \frac{d}{dx}\right)^{l-m}[(1-x)^{l-n}(1+x)^{l+n}],
\]
and the normalization constant  $c_{mn}^l$ is given
\[c_{mn}^l=2^{-l}\frac{(-1)^{l-n}i^{n-m}}{\sqrt{(l-n)!(l+n)!}}\sqrt{\frac{(l+m)!}{(l-m)!}}.\]
(see, e.g., \cite{ruzhansky2009pseudo}).

\begin{defn}
The matrix $(t^l_{mn})_{m,n}$, with indices $m, n$ satisfying $-l \leq m, n \leq l$ and $l - m, l - n \in \mathbb{Z}$, is denoted by $t^l_{mn}$. The standard convention $0! = 1$ is used throughout all expressions.
\end{defn}
According to the Peter–Weyl theorem\eqref{PETER}, the functions $t^l_{nm}$ form an orthonormal basis of $\ell^2(\text{SU}(2))$\cite[p. 629]{ruzhansky2009pseudo}. These functions are defined for $l \in \frac{1}{2} \mathbb{N}_0$, satisfying $-l \leq m, n \leq l$ and $l - m, l - n \in \mathbb{Z}$. Their explicit expression is given by
\[
t^l_{nm}(\omega(\phi, \theta, \psi)) = e^{-i(n\phi + m\psi)} P^l_{nm}(\cos(\theta)),
\]
where $P^l_{nm}(x)$ is defined as
\begin{align}\label{Expresation P_mn^l}
    P_{nm}^l(x)=c^l_{nm}\frac{(1-x)^{(m-n)/2}}{(1+x)^{(n+m)/2}} \left( \frac{d}{dx}\right)^{l-n}[(1-x)^{l-m}(1+x)^{l+m}].
\end{align}
and the normalization constant $c^l_{nm}$ is given
\begin{align*}
    c^l_{nm}=2^{-l}\frac{(-1)^{l-m}i^{m-n}}{\sqrt{(l-m)!(l+m)!}}\sqrt{\frac{(l+n)!}{(l-n)!}} .
\end{align*}
(see, e.g., \cite{ruzhansky2009pseudo}).\\%

It is important to note that, in this context, the ordering of the indices $(m, n)$ has been changed to $(n, m)$ compared to the original formulation of the Peter–Weyl theorem (cf. \cite{ruzhansky2009pseudo}). This modification allows to get a correct definition of  the matrix multiplication between $\hat{f}(l)_{mn}$ and $t^l_{nm}(x)$ in the Fourier coefficients of the equation.\\%

Therefore, the collection
\[
\left\{\sqrt{2l + 1} \; t^l_{nm} : l \in \tfrac{1}{2} \mathbb{N}_0, \; -l \leq m, n \leq l, \; l - m, l - n \in \mathbb{Z} \right\}
\]
forms an orthonormal basis of $\ell^2(\text{SU(2)})$ (cf. \cite{Preli,ruzhansky2009pseudo}). Hence, by the Peter–Weyl theorem\eqref{PETER}, any function $f \in \ell^2(\text{SU(2)})$ can be expressed in terms of the basis functions $t^l_{nm}(x)$ associated with the irreducible representations $l$ of $\text{SU(2)}$
\[
f(x) = \sum_{l \in \frac{1}{2}\mathbb{N}_0} (2l + 1) \sum_{m, n} \hat{f}(l)_{mn} \; t^l_{nm}(x),
\]
where the \textit{Fourier coefficients} $\hat{f}(l)_{mn}$ are given by
\[
\hat{f}(l)_{mn} = \int_{\text{SU(2)}} f(x) \; \overline{t^l_{nm}(x)} \, dx = \langle f, t^l_{nm} \rangle_{\ell^2(\text{SU(2)})}.
\]
(see, e.g., \cite{ruzhansky2009pseudo}).\\%

The basis functions $t^l_{nm}$ satisfy the orthonormality relation
\[
\langle t^l_{nm}, t^{l}_{n'm'} \rangle_{\ell^2(\text{SU(2)})} = \frac{\delta_{mm'} \delta_{nn'}}{2l + 1}.
\]%
(see, e.g., \cite{ruzhansky2009pseudo,Preli}).\\

We are now ready to define  the Fourier Transform on   $\text{SU(2)}$.                   
\begin{defn}
    The \textit{Fourier transform of }$f$ on $\text{SU(2)}$ is the sequence of matrices $\{ \hat{f}(l) \}_{l \in \frac{1}{2} \mathbb{N}_0}$, where each entry of each matrix $\hat{f}(l) \in \mathbb{C}^{(2l+1) \times (2l+1)}$ is defined by
\[
\hat{f}(l)_{mn} := \int_{\text{SU(2)}} f(x) \; \overline{t^l_{nm}(x)} \, dx = \langle f, t^l_{nm} \rangle_{\ell^2(\text{SU(2)})},
\]
for $-l \leq m, n \leq l$, and $l - m, l - n \in \mathbb{Z}$. 
\end{defn}

Let us provide a geometric interpretation of the FT on $\text{SU}(2)$. Recall that in the classical setting, the FT of a function defined over space-time allows its decomposition as a sum of simple waves (an orthonormal basis), whose image is the frequency  (cf. \cite{stein2011fourier}).\\

In the case of $\text{SU}(2)$, since this group is diffeomorphic to a complex 3-sphere $ \mathbb{S}^3$ embedded in $\mathbb{R}^4$ (cf. \cite{Faraut_2008}), any signal defined on this space can likewise be decomposed. The Fourier transform on SU(2) expresses any signal over complex 3-sphere embedded in $\mathbb{R}^4$ as a sum of generalized \textquotedblleft waves\textquotedblright, specifically in this context the functions $t^l_{nm}$, (an orthonormal basis) they are the generalized \textquotedblleft waves\textquotedblright \text{ }over complex 3-sphere embedded in $\mathbb{R}^4$.\\

The image of the FT on SU(2) is no longer a vector of coefficients, but rather a sequence of matrices, which can be interpreted as a \textquotedblleft frequency spectrum\textquotedblright . This FT can be conceptually represented as follows

\begin{figure}[H]
    \centering
    \includegraphics[width=0.5\linewidth]{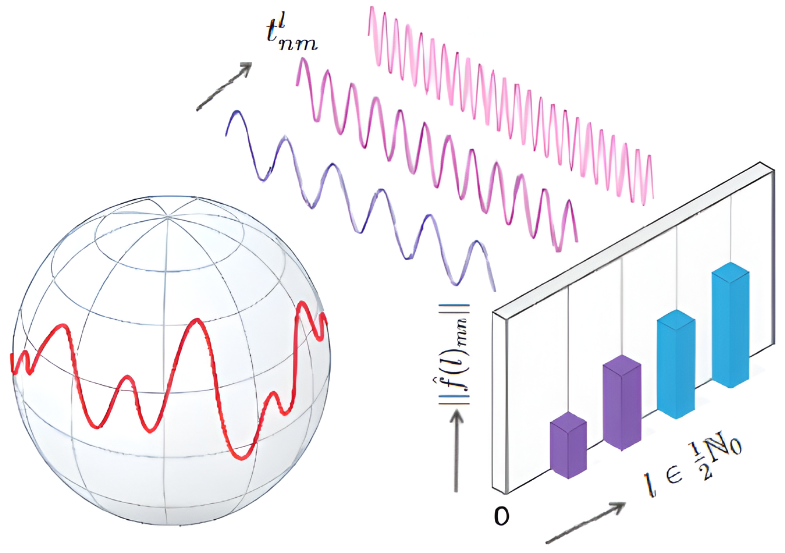}
    \caption{Representation of the FT on SU(2)}
    \label{fig:placeholder}
\end{figure}

This figure illustrates the conceptual analogy between the classical Fourier transform and its counterpart on $\text{SU}(2)$. On the left, we consider a function defined over complex 3-sphere $ \mathbb{S}^3$ embedded in $\mathbb{R}^4$. Through the Fourier transform on $\text{SU}(2)$, this function is decomposed as a sum of generalized \textquotedblleft waves\textquotedblright \text{ }$t^{l}_{nm}$,  which form an orthonormal basis for the space.

On the right, we display a plane where the horizontal axis corresponds to the half-integer parameter $l \in \tfrac{1}{2} \mathbb{N}_0$, and the vertical axis represents the norm of the matrix $\hat{f}(l)_{mn}$ associated with each $l$, i.e., $\|\hat{f}(l)_{mn}\|_F$ (Frobenius norm). This results in a \textquotedblleft frequency spectrum\textquotedblright \text{ }representation given by a sequence of matrices indexed by $l$, capturing the spectral content of the original signal on $\text{SU}(2)$.%

\section{Fast Fourier Transform on SU(2)}\label{sec:FFTSU2}

The objective of this section is to develop an explicit Fast Fourier Transform (FFT) algorithm for the group SU(2) that follows the classical Cooley-Tukey divide-and-conquer scheme. The core idea is to reformulate the Fourier transform on SU(2) in a way that enables the systematic application of FFT techniques, leading to a significant reduction in computational cost compared to the direct approach.

Our construction is inspired by existing FFT methods on the 2-Sphere, in particular those presented in \textit{FFTs for the 2-sphere—Improvements and Variations} (cf. \cite{2-sphere}), as well as by the Quantum Fourier Transform framework developed for SU(2) (cf. \cite{Quantum_SU(2)}). 
While drawing on these ideas, the algorithm presented here is developed in a purely classical setting and tailored specifically to the classical Fourier transform on SU(2).

For a broader theoretical perspective on fast Fourier transforms on compact groups, we refer the reader to \textit{Efficient computation of Fourier transforms on compact groups} (cf. \cite{FFT_Compact_Groups}).

Before proceeding formally with the derivation of the algorithm, we will provide a high-level overview of the intended approach. The following scheme will guide our development
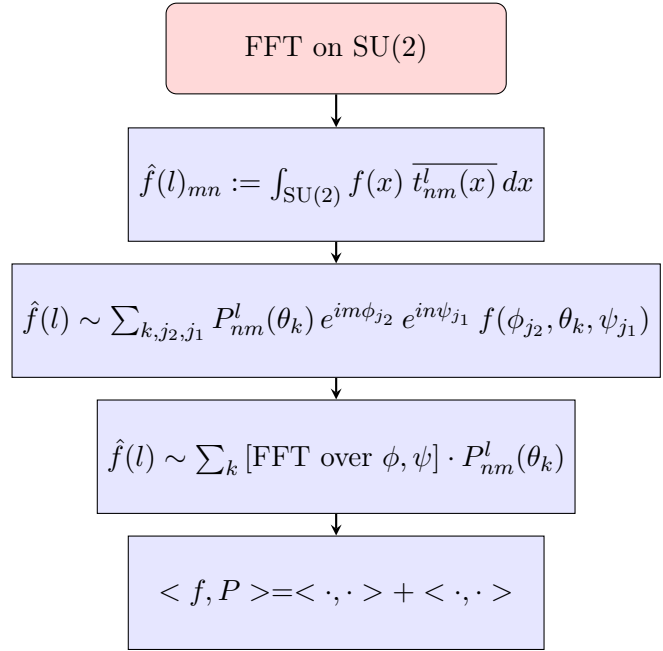
\begin{figure}[H]
\centering
\begin{tikzpicture}[node distance=1.8cm]

\tikzstyle{startstop} = [rectangle, rounded corners, minimum width=4.5cm, minimum height=1.2cm, text centered, draw=black, fill=red!15]
\tikzstyle{process} = [rectangle, minimum width=5.5cm, minimum height=1.5cm, text centered, draw=black, fill=blue!10]
\tikzstyle{arrow} = [thick,->,>=stealth]

\node (start2) [startstop] {FFT on $\text{SU}(2)$};
\node (step1) [process, below of=start2] {
$\hat{f}(l)_{mn} := \int_{\text{SU(2)}} f(x) \; \overline{t^l_{nm}(x)} \, dx$
};
\node (step2) [process, below of=step1] {
$\hat{f}(l) \sim \sum_{k,j_2,j_1} P^l_{nm}(\theta_k) \, e^{i m\phi_{j_2} } \, e^{i n \psi_{j_1}} \, f(\phi_{j_2}, \theta_k, \psi_{j_1})$
};
\node (cond2) [process, below of=step2] {
\(\hat{f}(l) \sim \sum_k \left[\text{FFT over } \phi, \psi \right] \cdot P^l_{nm}(\theta_k)\)
};
\node (result2) [process, below of=cond2] {
$<f,P>=<\cdot,\cdot>+<\cdot,\cdot>$
};



\draw [arrow] (start2) -- (step1);
\draw [arrow] (step1) -- (step2);
\draw [arrow] (step2) -- (cond2);
\draw [arrow] (cond2) -- (result2);

\end{tikzpicture}
\caption{General Scheme for the FFT on SU(2)}
\label{fig:fft_su2}
\end{figure}
For the development of the FFT on $\text{SU}(2)$, a \textit{divide-and-conquer} strategy was adopted in order to preserve the classical structure of the FFT. To this end, the algorithm was organized into four main steps that guide the derivation
\begin{enumerate}
    \item  \textbf{Formal definition:} The process begins with the rigorous definition of the Fourier transform on $\text{SU}(2)$, based on the irreducible representations of the group.

    \item  \textbf{Discretization:} Next, the domain is discretized appropriately, allowing the replacement of integrals with finite sums, which is more suitable for computational purposes.

    \item \textbf{Application of 2D FFT:} A two-dimensional FFT is applied to the angular variables $\phi$ and $\psi$, exploiting their appearance as complex exponentials.

    \item \textbf{Decomposition of the inner product:} Finally, recursive properties of the Jacobi polynomials are used to split the inner product into two independent components, completing the divide-and-conquer structure.
\end{enumerate}
This scheme enables an efficient implementation of the algorithm while maintaining the essence of the classical FFT in the non-commutative setting of $\text{SU}(2)$.

Now that we have a clear conceptual outline, we proceed with the rigorous derivation of the algorithm.

Let us remember that, the \textit{Fourier transform of }$f$ on $\text{SU(2)}$ is the sequence of matrices $\{ \hat{f}(l) \}_{l \in \frac{1}{2} \mathbb{N}_0}$, where each entry of each matrix $\hat{f}(l) \in \mathbb{C}^{(2l+1) \times (2l+1)}$ is defined by
\[
\hat{f}(l)_{mn} := \int_{\text{SU(2)}} f(x) \; \overline{t^l_{nm}(x)} \, dx = \langle f, t^l_{nm} \rangle_{\ell^2(\text{SU(2)})},
\]
for $-l \leq m, n \leq l$, and $l - m, l - n \in \mathbb{Z}$.

Using Euler angles is defined as
\begin{align*}
    \hat{f}(l)_{mn} := \frac{1}{8\pi^2} \int_{-\pi}^{\pi} \int_0^{\pi} \int_{-\pi}^{\pi} f(\phi, \theta, \psi) \, \overline{t^l_{nm}(\phi, \theta, \psi)} \, \sin(\theta) \, d\phi \, d\theta \, d\psi,
\end{align*}
where the orthonormal basis $t_{nm}^l(\phi, \theta, \psi)$, for $l \in \frac{1}{2} \mathbb{N}_0$, $-l \leq m, n \leq l$, and $l - m, l - n \in \mathbb{Z}$, is defined as
\begin{align*}
    t_{nm}^l(\phi, \theta, \psi) = P_{nm}^l(\cos(\theta)) e^{-i(n\phi + m\psi)}.
\end{align*}
Therefore
\begin{align*}
    \hat{f}(l)_{mn} := \frac{1}{8\pi^2}\int_{-\pi}^{\pi} \int_0^{\pi} \int_{-\pi}^{\pi} f(\phi, \theta, \psi) \;P_{nm}^l(\cos(\theta)) e^{i(n\phi + m\psi)}  \sin(\theta) \, d\phi \, d\theta \, d\psi.
\end{align*}
(see, e.g., \cite{ruzhansky2009pseudo}).

The goal is to bring the Fourier transform on $\text{SU(2)}$ into a framework where FFT techniques can be applied to improve computational efficiency. We say that a function $f \in \ell^2(\text{SU(2)})$ is \textit{band-limited} with \textit{band-limit} or \textit{bandwidth} $N \geq 0$ if $\hat{f}(l)_{mn} = 0$ for all $l \geq N$.

Therefore, we begin by uniformly discretizing the angles $\phi$, $\psi$, and $\theta$ in order to leverage the FFT algorithm in the corresponding summations.

\begin{defn}
Let $N \in \mathbb{N}$ be the bandwidth. We define the sampling grid on SU(2) as the set of points $(\phi_{j_1}, \theta_k, \psi_{j_2})$ where:
\begin{align*}
\phi_{j_1} &= -\pi + j_1 \frac{2\pi}{N-1}, \quad j_1 = 0, \dots, N-1, \\
\psi_{j_2} &= -\pi + j_2 \frac{2\pi}{N-1}, \quad j_2 = 0, \dots, N-1, \\
\theta_k &= k \frac{\pi}{N-1}, \quad k = 0, \dots, N-1.
\end{align*}
\end{defn}

Using this discretization, the Fourier coefficient $\hat{f}(l)_{mn}$ can be approximated by a Riemann sum.

\begin{prop}
The computation of the Fourier coefficient $\hat{f}(l)_{mn}$ on the grid factorizes as
\begin{equation}
\hat{f}(l)_{mn} \approx \frac{1}{8\pi^2} \sum_{k=0}^{N-1} P^l_{nm}(\cos\theta_k) \sin\theta_k \left[ \sum_{j_1=0}^{N-1} e^{in\phi_{j_1}} \sum_{j_2=0}^{N-1} e^{im\psi_{j_2}} f(\phi_{j_1}, \theta_k, \psi_{j_2}) \right].
\end{equation}
The term in brackets corresponds to a two-dimensional Discrete Fourier Transform (DFT) with respect to indices $j_1$ and $j_2$.
\end{prop}

To derive this expression, we substitute the continuous integrals with finite Riemann sums over the uniform grid. Crucially, the kernel function separates the angular dependencies, allowing us to group the terms involving $\phi$ and $\psi$
\begin{equation*}
    \underbrace{\sum_{j_1=0}^{N-1} \sum_{j_2=0}^{N-1} f(\phi_{j_1}, \theta_k, \psi_{j_2}) \, e^{in\phi_{j_1}} e^{im\psi_{j_2}}}_{\text{2D Discrete Fourier Transform}}.
\end{equation*}
We observe that for each fixed $\theta_k$, this inner summation is exactly a 2D DFT with respect to the indices $(j_1, j_2)$. 

Therefore, instead of performing a naive summation, we can compute this term efficiently using a 2D FFT routine. Letting $f_2(\theta_k)$ denote the output of this FFT step for the $k$-th, the final calculation of $\hat{f}(l)_{mn}$ reduces to the one-dimensional contraction over $\theta_k$:
\[
    \hat{f}(l)_{mn} \approx \frac{1}{8\pi^2} \sum_{k=0}^{N-1} f_2(\theta_k) \sin(\theta_k) P^l_{nm}(\cos\theta_k).
\]

We can view $f_2(\theta_k) \sin(\theta_k)$ as a sequence of values indexed by $k = 0, 1, \ldots, N - 1$.

The computation is completed by performing the required discrete Legendre transforms, defined as a set of sums
\begin{equation}\label{Core_eq}
    \hat{f}(l)_{mn}\approx\sum_{k=0}^{N-1}[s]_k P^l_{nm}(\cos \theta_k) = \langle \mathbf{s}, \mathbf{P}^l_{nm} \rangle
\end{equation}
where $[s]_k = f_2(\theta_k) \sin(\theta_k)$ is an arbitrary input vector $\mathbf{s}$ with its $k^{\text{th}}$ component equal to $[s]_k$ (for notational simplicity, we suppress the factor $1/8\pi^2$).

In this context, we use the notation for the discrete inner product and define $\mathbf{P}^l_{nm}$ as the vector formed by evaluating the Legendre polynomial $P^l_{nm}(\cos \theta_k)$ at the corresponding points $\cos \theta_k$
\begin{align*}
    \mathbf{P}_{nm}^l=\begin{pmatrix}
        P_{nm}^l(\cos\theta_0)&\\
        \vdots\\
        P_{nm}^l(\cos\theta_{N-1})
    \end{pmatrix}.
\end{align*}

The core of our contribution is the efficient computation of \eqref{Core_eq}. For clarity, we restrict our derivation to the case $m=n=0$, where $P^l_{00} = P_l$ are the Legendre polynomials. The generalization to Jacobi polynomials follows an analogous recursive structure. We exploit the three-term recurrence relation of Legendre polynomials to derive a divide-and-conquer strategy.

\begin{lem}
For any  $l = 0,1,\ldots,N-1$. We can express $P_{l+1}(x)$ in terms of $P_{l-1}(x)$ and $P_l(x)$:
\begin{align}\label{eq:P_l}
    P_{l+1}(\cos\theta) = \frac{2l+1}{l+1} \cos\theta P_l(\cos\theta) - \frac{l}{l+1} P_{l-1}(\cos\theta).
\end{align}
With initial conditions
$P_0(\cos\theta) = 1$ and $P_{-1}(\cos\theta) = 0$,
\end{lem}

\begin{proof}
    First, recall from Equation \eqref{Expresation P_mn^l} that for the case $m=n=0$, the generalized function $P^l_{00}(x)$ reduces to the standard Legendre polynomial, denoted here as $P_l(x)$. According to the recurrence relation for the \textit{Legendre polynomials} $P_n(x)$ is given by
\begin{align}\label{eq:rodrigues}
(n+1) P_{n+1}(x) - (2n+1)x P_n(x) + n P_{n-1}(x) = 0.
\end{align}
    
    It is a classical result that polynomials defined by \eqref{eq:rodrigues} satisfy the recurrence relation \eqref{eq:P_l}
    \begin{equation} 
        (l+1) P_{l+1}(x) - (2l+1)x P_l(x) + l P_{l-1}(x) = 0.
    \end{equation}
    
    By setting $x = \cos\theta$ and isolating the term $P_{l+1}(x)$, we obtain the expression stated in the lemma
    \[
        P_{l+1}(\cos\theta) = \frac{2l+1}{l+1} \cos\theta P_l(\cos\theta) - \frac{l}{l+1} P_{l-1}(\cos\theta).
    \]
    
    Finally, we verify the initial conditions. For $l=0$, the formula yields
    \[
        P_1(\cos\theta) = \frac{1}{1} \cos\theta P_0(\cos\theta) - 0.
    \]
    Using \eqref{eq:rodrigues} for $l=0$, we have $P_0(x) = 1$, which implies $P_1(x) = x$. This is consistent with the standard definition. By convention, we set $P_{-1}(\cos\theta) = 0$ to maintain the validity of the recurrence at the boundary $l=0$.
\end{proof}

The previous lemma establishes the recurrence relation for sequential steps. However, in the context of a recursive algorithm, it is often necessary to evaluate polynomials at a depth $L+r$ relative to an arbitrary base level $L$, rather than iterating from the origin $l=0$. This motivates the following result, which generalizes the recurrence to a shift of $r$ steps.

\begin{prop}
For any level $L$ and step $r \geq 1$, the polynomial $P_{L+r}$ can be expressed as a linear combination of $P_L$ and $P_{L-1}$
\begin{equation}\label{CAP6-8}
P_{L+r}(\cos\theta) = A_r^L(\cos\theta) P_L(\cos\theta) + B_r^L(\cos\theta) P_{L-1}(\cos\theta),
\end{equation}
With initial conditions $A_0^L=1, B_0^L=0, A_{-1}^L=0, B_{-1}^L=1$.
\end{prop}

\begin{proof}
Fixing a level $L$ and iterating the three-term recurrence relation \eqref{eq:P_l} forward $r$ steps, we obtain trigonometric polynomials $A_r^L(\cos\theta)$ and $B_r^L(\cos\theta)$ such that, for all $r \geq 1$,
\begin{equation}\label{eq-eq}
P_{L+r}(\cos\theta)
= A_r^L(\cos\theta)\, P_L(\cos\theta)
+ B_r^L(\cos\theta)\, P_{L-1}(\cos\theta).
\end{equation}

This representation follows directly from the recurrence relation
\[
(L+r) P_{L+r}(\cos\theta)
- (2(L+r)-1)\cos\theta\, P_{L+r-1}(\cos\theta)
+ (L+r-1) P_{L+r-2}(\cos\theta) = 0,
\]
which expresses each polynomial $P_{L+r}$ in terms of its two immediate predecessors.\\

To determine the initial conditions for the coefficients $A_r^L$ and $B_r^L$, we consider the lowest values of $r$. Although the recurrence is defined for $r \geq 1$, it is convenient to extend it formally to $r = -1$ and $r=0$ in order to complete the set of initial conditions. \\

For $r = 0$, equation \eqref{eq-eq} reduces to
\[
P_L(\cos\theta)
= A_0^L\, P_L(\cos\theta)
+ B_0^L\, P_{L-1}(\cos\theta),
\]
which must hold for any $L$. This immediately implies
\[
A_0^L = 1, \qquad B_0^L = 0.
\]

For $r = -1$, we obtain
\[
P_{L-1}(\cos\theta)
= A_{-1}^L\, P_L(\cos\theta)
+ B_{-1}^L\, P_{L-1}(\cos\theta),
\]
which is satisfied for all $L$ only if
\[
A_{-1}^L = 0, \qquad B_{-1}^L = 1.
\]

Therefore, the coefficients $A_r^L$ and $B_r^L$ are initialized by
\[
A_0^L = 1, \quad B_0^L = 0, \quad
A_{-1}^L = 0, \quad B_{-1}^L = 1,
\]
ensuring consistency of the representation \eqref{eq-eq} at the base levels.
    
\end{proof}
The next step is to derive a recurrence formula for the (shifted) Legendre polynomials $A_r^L$ and $B_r^L$.

\begin{prop}
Let $L$ be fixed. Assume that for each $r \geq 1$ the Legendre polynomials satisfy
\[
P_{L+r}(\cos\theta)
= A_r^L(\cos\theta)\, P_L(\cos\theta)
+ B_r^L(\cos\theta)\, P_{L-1}(\cos\theta),
\]
with initial conditions
\[
A_0^L = 1, \quad B_0^L = 0, \qquad
A_{-1}^L = 0, \quad B_{-1}^L = 1.
\]
Then the coefficients $A_r^L$ and $B_r^L$ satisfy the matrix recurrence
\[
\begin{pmatrix}
A_{r+1}^L & B_{r+1}^L \\
A_r^L & B_r^L
\end{pmatrix}
=
\begin{pmatrix}
\dfrac{2L+2r+1}{L+r+1}\cos\theta &
-\dfrac{L+r}{L+r+1} \\
1 & 0
\end{pmatrix}
\begin{pmatrix}
A_r^L & B_r^L \\
A_{r-1}^L & B_{r-1}^L
\end{pmatrix},
\]
or, equivalently, the scalar recurrences
\[
\begin{cases}
A_{r+1}^L
= \dfrac{2L+2r+1}{L+r+1}\cos\theta\, A_r^L
- \dfrac{L+r}{L+r+1} A_{r-1}^L, \\[0.8em]
B_{r+1}^L
= \dfrac{2L+2r+1}{L+r+1}\cos\theta\, B_r^L
- \dfrac{L+r}{L+r+1} B_{r-1}^L.
\end{cases}
\]
\end{prop}

\begin{proof}
We start from the three-term recurrence relation for Legendre polynomials,
\[
(L+r+1) P_{L+r+1}(\cos\theta)
= (2L+2r+1)\cos\theta\, P_{L+r}(\cos\theta)
- (L+r) P_{L+r-1}(\cos\theta).
\]
Dividing by $L+r+1$ and writing this relation in matrix form yields
\begin{equation}\label{MATRIZ_LEGEN1}
\begin{pmatrix}
P_{L+r+1} \\
P_{L+r}
\end{pmatrix}
=
\begin{pmatrix}
\dfrac{2L+2r+1}{L+r+1}\cos\theta &
-\dfrac{L+r}{L+r+1} \\
1 & 0
\end{pmatrix}
\begin{pmatrix}
P_{L+r} \\
P_{L+r-1}
\end{pmatrix}.
\end{equation}

By hypothesis, each polynomial $P_{L+k}$ can be expressed as
\[
P_{L+k}
= A_k^L P_L + B_k^L P_{L-1}.
\]
Substituting the expressions for $P_{L+r}$ and $P_{L+r-1}$ into
\eqref{MATRIZ_LEGEN1}, we obtain
\begin{equation}\label{EQmatrix1}
\begin{pmatrix}
P_{L+r+1} \\
P_{L+r}
\end{pmatrix}
=
\begin{pmatrix}
\dfrac{2L+2r+1}{L+r+1}\cos\theta &
-\dfrac{L+r}{L+r+1} \\
1 & 0
\end{pmatrix}
\begin{pmatrix}
A_r^L & B_r^L \\
A_{r-1}^L & B_{r-1}^L
\end{pmatrix}
\begin{pmatrix}
P_L \\
P_{L-1}
\end{pmatrix}.
\end{equation}

On the other hand, applying the representation of $P_{L+r+1}$ directly gives
\[
\begin{pmatrix}
P_{L+r+1} \\
P_{L+r}
\end{pmatrix}
=
\begin{pmatrix}
A_{r+1}^L & B_{r+1}^L \\
A_r^L & B_r^L
\end{pmatrix}
\begin{pmatrix}
P_L \\
P_{L-1}
\end{pmatrix}.
\]

Since the vectors $(P_L, P_{L-1})^{\mathsf T}$ are the same in both expressions,
equating the corresponding coefficient matrices in \eqref{EQmatrix1} yields
\[
\begin{pmatrix}
A_{r+1}^L & B_{r+1}^L \\
A_r^L & B_r^L
\end{pmatrix}
=
\begin{pmatrix}
\dfrac{2L+2r+1}{L+r+1}\cos\theta &
-\dfrac{L+r}{L+r+1} \\
1 & 0
\end{pmatrix}
\begin{pmatrix}
A_r^L & B_r^L \\
A_{r-1}^L & B_{r-1}^L
\end{pmatrix}.
\]
Comparing the entries of both sides gives the stated recurrence relations for
$A_{r+1}^L$ and $B_{r+1}^L$.
\end{proof}

Using a similar argument, we now aim to derive a recurrence formula for the Legendre polynomials shifted $s$ steps forward starting from $r$ steps. The core idea is that progressing $r + s$ steps from level $L$ is equivalent to
\begin{enumerate}
\item Advancing $r$ steps from $L$ (reaching $L + r$).
\item Then advancing $s$ additional steps from $L + r$ (reaching $L + r + s$).
\end{enumerate}

This implies that the matrix representing a shift of $r + s$ steps can be expressed as the composition (i.e., matrix product) of the matrices representing a shift of $r$ steps and $s$ steps, respectively.

\begin{prop}\label{Prop-P_l}
Let $L$ be fixed and let $A_r^L(\cos\theta)$ and $B_r^L(\cos\theta)$ be defined by
\[
P_{L+r}(\cos\theta)
= A_r^L(\cos\theta)\, P_L(\cos\theta)
+ B_r^L(\cos\theta)\, P_{L-1}(\cos\theta),
\qquad r \ge 1.
\]
Then, for any integers $r,s \ge 0$, the coefficients satisfy the composition rule
\begin{equation}\label{eq:composition}
\begin{pmatrix}
A_{r+s}^L & B_{r+s}^L \\
A_{r+s-1}^L & B_{r+s-1}^L
\end{pmatrix}
=
\begin{pmatrix}
A_s^{L+r} & B_s^{L+r} \\
A_{s-1}^{L+r} & B_{s-1}^{L+r}
\end{pmatrix}
\begin{pmatrix}
A_r^L & B_r^L \\
A_{r-1}^L & B_{r-1}^L
\end{pmatrix}.
\end{equation}
Moreover, the initial conditions are
\[
A_r^0 = P_r, \qquad B_r^0 = 0, \qquad r \ge 1.
\]
\end{prop}

\begin{proof}
The key observation is that advancing $r+s$ steps from level $L$ can be decomposed
into two successive shifts: first $r$ steps from $L$ to $L+r$, and then $s$ steps
from $L+r$ to $L+r+s$.\\

By definition, the matrix
\[
\begin{pmatrix}
A_r^L & B_r^L \\
A_{r-1}^L & B_{r-1}^L
\end{pmatrix}
\]
represents the transformation that maps $(P_L, P_{L-1})^{\mathsf T}$ to
$(P_{L+r}, P_{L+r-1})^{\mathsf T}$.
Similarly, advancing $s$ steps from level $L+r$ is represented by
\[
\begin{pmatrix}
A_s^{L+r} & B_s^{L+r} \\
A_{s-1}^{L+r} & B_{s-1}^{L+r}
\end{pmatrix}.
\]

Applying first the shift of $r$ steps and then the shift of $s$ steps yields
\begin{equation}\label{eq:rs-product}
\begin{pmatrix}
A_s^{L+r} & B_s^{L+r} \\
A_{s-1}^{L+r} & B_{s-1}^{L+r}
\end{pmatrix}
\begin{pmatrix}
A_r^L & B_r^L \\
A_{r-1}^L & B_{r-1}^L
\end{pmatrix}
=
\begin{pmatrix}
A_s^{L+r}A_r^L + B_s^{L+r}A_{r-1}^L &
A_s^{L+r}B_r^L + B_s^{L+r}B_{r-1}^L \\
A_{s-1}^{L+r}A_r^L + B_{s-1}^{L+r}A_{r-1}^L &
A_{s-1}^{L+r}B_r^L + B_{s-1}^{L+r}B_{r-1}^L
\end{pmatrix}.
\end{equation}

On the other hand, advancing $r+s$ steps directly from level $L$ is represented by
\[
\begin{pmatrix}
A_{r+s}^L & B_{r+s}^L \\
A_{r+s-1}^L & B_{r+s-1}^L
\end{pmatrix}.
\]
Since both procedures map $(P_L, P_{L-1})^{\mathsf T}$ to
$(P_{L+r+s}, P_{L+r+s-1})^{\mathsf T}$, the corresponding matrices must coincide.
Equating the matrices in \eqref{eq:rs-product} with the above one yields
the identity \eqref{eq:composition}, or equivalently the scalar relations
\[
\begin{cases}
A_{r+s}^L = A_s^{L+r}A_r^L + B_s^{L+r}A_{r-1}^L, \\[0.3em]
B_{r+s}^L = A_s^{L+r}B_r^L + B_s^{L+r}B_{r-1}^L, \\[0.3em]
A_{r+s-1}^L = A_{s-1}^{L+r}A_r^L + B_{s-1}^{L+r}A_{r-1}^L, \\[0.3em]
B_{r+s-1}^L = A_{s-1}^{L+r}B_r^L + B_{s-1}^{L+r}B_{r-1}^L.
\end{cases}
\]

The consistency with the initial conditions is immediate. If $r=0$ or $s=0$, the
corresponding matrix reduces to the identity, so that \eqref{eq:composition}
holds trivially.

Finally, for $L=0$ the defining relation becomes
\[
P_r(\cos\theta)
= A_r^0(\cos\theta)\, P_0(\cos\theta)
+ B_r^0(\cos\theta)\, P_{-1}(\cos\theta).
\]
By convention, Legendre polynomials are undefined for negative indices, so we assume $P_{-1}(\cos\theta) = 0$, and using $P_0=1$ we obtain
\[
P_r(\cos\theta) = A_r^0(\cos\theta),
\]
which implies
\[
A_r^0 = P_r, \qquad B_r^0 = 0,
\]
for all $r \ge 1$. This completes the proof.
\end{proof}

The results developed in the previous \eqref{Prop-P_l} provide a recursive and algebraic
description of how Legendre polynomials of higher degree can be generated from
lower-degree ones through the shifted coefficients $A_r^L$ and $B_r^L$.
Beyond their theoretical interest, these coefficients play a central role in the
efficient evaluation of projections onto high-degree Legendre polynomials.

The purpose of introducing the shifted polynomials $A_r^L$ and $B_r^L$ is to allow us to express projections onto higher-degree Legendre polynomials as sums of projections onto (shifted) Legendre polynomials of lower degree. Suppose that, during the computation of the inner products $\langle \mathbf{s}, \mathbf{P}_j \rangle$ (for $j \leq L$), we had already stored the components of the vectors $\mathbf{s}^j = \mathbf{s} \cdot \mathbf{P}_j$, defined as the pointwise product between the vectors $\mathbf{s}$ and $\mathbf{P}_j$. That is, for each point $k$, we multiply the value $s_k$ (the function $s$ evaluated at $\theta_k$) by the value $P_j(\cos\theta_k)$ (the Legendre polynomial of degree $j$ evaluated at $\cos\theta_k$)
\[
\mathbf{s}^j = \mathbf{s} \cdot \mathbf{P}_j, \quad \text{with} \quad (\mathbf{s} \cdot \mathbf{P}_j)_k = s_k \cdot P_j(\cos\theta_k).\]
In other words, $\mathbf{s}^j$ is a vector whose $k$-th component is the product of the corresponding values of $s$ and $P_j$ evaluated at the nodes $\theta_k$. The key observation is that these precomputed quantities can be reused to obtain
projections onto higher-degree polynomials without explicitly evaluating
$P_{L+r}$.

The following theorem formalizes this idea.
It shows that the projection onto $P_{L+r}$ can be written as a sum of inner
products involving only the stored vectors $\mathbf{s}^L$ and
$\mathbf{s}^{L-1}$ and recursively defined shifted polynomials of lower degree.
Moreover, when the total shift $r$ is decomposed into smaller increments, the
resulting expression exhibits a natural divide-and-conquer structure, directly
analogous to the classical Cooley--Tukey decomposition in the FFT.

\begin{thm}[Divide-and-conquer]\label{thm:final_legendre_fft}
Let $L \ge 1$ be fixed and let $\mathbf{s}$ be a discrete function sampled at nodes
$\{\theta_k\}_{k=1}^N$.  
Assume that for all $j \le L$ the vectors
\[
\mathbf{s}^j := \mathbf{s} \cdot \mathbf{P}_j,
\qquad (\mathbf{s}^j)_k = s_k\, P_j(\cos\theta_k),
\]
have been precomputed and stored.  
Then, for any $r \ge 1$, the projection of $\mathbf{s}$ onto the Legendre polynomial
$P_{L+r}$ satisfies
\begin{equation}\label{eq:basic_projection}
\langle \mathbf{s}, \mathbf{P}_{L+r} \rangle
=
\langle \mathbf{s}^L, \mathbf{A}_r^L \rangle
+
\langle \mathbf{s}^{L-1}, \mathbf{B}_r^L \rangle,
\end{equation}
where $A_r^L$ and $B_r^L$ are the shifted Legendre polynomials defined by
\[
P_{L+r}
=
A_r^L P_L + B_r^L P_{L-1}.
\]

Moreover, if the total shift $r$ is decomposed as
\[
r = n_1 + n_2 + \cdots + n_k, \qquad n_i \in \mathbb{N},
\]
then the inner product $\langle \mathbf{s}, \mathbf{P}_{L+r} \rangle$ admits the
recursive expansion
\begin{equation}\label{eq:general_decomposition}
\langle \mathbf{s}, \mathbf{P}_{L+r} \rangle
=
\sum_{\tau_1 \in \{\mathbf{A},\mathbf{B}\}}
\sum_{\epsilon_1=0}^{1}
\sum_{\epsilon_2=0}^{1}
\cdots
\sum_{\epsilon_{k-1}=0}^{1}
\left\langle
\mathbf{s}^{L^*(\tau_1)},
\mathbf{V}_{\tau_1,\epsilon_1,\ldots,\epsilon_{k-1}}
\right\rangle,
\end{equation}
where
\begin{itemize}
\item $L_0 := L$ and $L_{i-1} := L + \sum_{j=1}^{i-1} n_j$ for $i \ge 1$;
\item
\[
\mathbf{s}^{L^*(\tau_1)} =
\begin{cases}
\mathbf{s}^L, & \tau_1 = \mathbf{A}, \\
\mathbf{s}^{L-1}, & \tau_1 = \mathbf{B};
\end{cases}
\]
\item the vector $\mathbf{V}_{\tau_1,\epsilon_1,\ldots,\epsilon_{k-1}}$ is defined by
\[
\mathbf{V}_{\tau_1,\epsilon_1,\ldots,\epsilon_{k-1}}
=
\mathbf{\tau_1}_{\,n_1-\epsilon_1}^{L}
\;\odot\;
\bigodot_{i=2}^{k}
\mathbf{\tau_i(\epsilon_{i-1})}_{\,n_i}^{L_{i-1}},
\]
with the selector
\[
\tau_i(\epsilon_{i-1}) =
\begin{cases}
\mathbf{A}, & \epsilon_{i-1}=0,\\
\mathbf{B}, & \epsilon_{i-1}=1.
\end{cases}
\]
\end{itemize}
All products $\odot$ are Hadamard (pointwise) products.
\end{thm}

\begin{proof}
The identity \eqref{eq:basic_projection} follows directly from the definition
\[
P_{L+r} = A_r^L P_L + B_r^L P_{L-1},
\]
together with linearity of the inner product and the definition
$\mathbf{s}^j = \mathbf{s}\cdot\mathbf{P}_j$.\\

Now assume that $r = n_1 + n_2$.  
Using the composition rule for the shifted coefficients,
\[
\begin{pmatrix}
A_{r}^L & B_{r}^L \\
A_{r-1}^L & B_{r-1}^L
\end{pmatrix}
=
\begin{pmatrix}
A_{n_2}^{L+n_1} & B_{n_2}^{L+n_1} \\
A_{n_2-1}^{L+n_1} & B_{n_2-1}^{L+n_1}
\end{pmatrix}
\begin{pmatrix}
A_{n_1}^L & B_{n_1}^L \\
A_{n_1-1}^L & B_{n_1-1}^L
\end{pmatrix},
\]
and inserting this into \eqref{eq:basic_projection}, we obtain
\[
\langle \mathbf{s}, \mathbf{P}_{L+r} \rangle
=
\langle \mathbf{s}^L, A_{n_2}^{L+n_1}A_{n_1}^L + B_{n_2}^{L+n_1}A_{n_1-1}^L \rangle
+
\langle \mathbf{s}^{L-1}, A_{n_2}^{L+n_1}B_{n_1}^L + B_{n_2}^{L+n_1}B_{n_1-1}^L \rangle.
\]
By the bilinearity of the inner product, this splits into four terms involving
Hadamard products of lower-degree shifted polynomials.\\

Iterating this argument for a decomposition
$r = n_1 + \cdots + n_k$ yields a binary branching at each step, producing
$2^{k-1}$ terms.  
Each term corresponds to a unique sequence
$(\tau_1,\epsilon_1,\ldots,\epsilon_{k-1})$ and has the form
\[
\left\langle
\mathbf{s}^{L^*(\tau_1)},
\mathbf{\tau_1}_{\,n_1-\epsilon_1}^{L}
\odot
\bigodot_{i=2}^{k}
\mathbf{\tau_i(\epsilon_{i-1})}_{\,n_i}^{L_{i-1}}
\right\rangle,
\]
which gives exactly the expansion \eqref{eq:general_decomposition}.\\

This establishes a divide-and-conquer structure for the computation of
$\langle \mathbf{s}, \mathbf{P}_{L+r} \rangle$, completing the proof.
\end{proof}

\begin{ex}

Let \( k = 2 \). Consider a total advance of \( r = n_1 + n_2 \) steps from level \( L \). The generalized inner product expression decomposes into 4 terms. \\

The expression becomes
\[
\langle \mathbf{s}, \mathbf{P}_{L+r} \rangle=\langle \mathbf{s}^L, \mathbf{A}_r^L \rangle + \langle \mathbf{s}^{L-1}, \mathbf{B}_r^L \rangle = \sum_{\tau_1 \in \{\mathbf{A},\mathbf{B}\}} \sum_{\epsilon_1=0}^{1} \left\langle \mathbf{s}^{L^*(\tau_1)}, \mathbf{V}_{\tau_1, \epsilon_1} \right\rangle,
\]
expanded as
\[
= \left\langle \mathbf{s}^{L^*(\mathbf{A})}, \mathbf{V}_{\mathbf{A}, 0} \right\rangle + \left\langle \mathbf{s}^{L^*(\mathbf{A})}, \mathbf{V}_{\mathbf{A}, 1} \right\rangle + \left\langle \mathbf{s}^{L^*(\mathbf{B})}, \mathbf{V}_{\mathbf{B}, 0} \right\rangle + \left\langle \mathbf{s}^{L^*(\mathbf{B})}, \mathbf{V}_{\mathbf{B}, 1} \right\rangle.
\]
Solving the first term
\[
\left\langle \mathbf{s}^{L^*(\mathbf{A})}, \mathbf{V}_{\mathbf{A}, 0} \right\rangle .
\]
For \( \tau_1 = \mathbf{A} \), \( \epsilon_1 = 0 \)
\begin{itemize}
    \item Compute
    $$\mathbf{s}^{L^*(\mathbf{A})} = \mathbf{s}^L \quad \text{(since $\tau_1 = \mathbf{A}$)}.$$
    
    \item Compute
    $$\mathbf{V}_{\mathbf{A}, 0} = \mathbf{\tau_2(\epsilon_1)}_{n_2}^{L_{1}} \odot \mathbf{\tau_1}_{n_1 - \epsilon_1}^{L},$$
    where \(\tau_2(\epsilon_1) = \mathbf{A}\) (since \(\epsilon_1 = 0\)) and \(L_1 = L + n_1\), so
    \[
    \mathbf{V}_{\mathbf{A},0} = \mathbf{A}_{n_2}^{L + n_1}  \mathbf{A}_{n_1}^{L}.
    \]
    Therefore
    \[
    \left\langle \mathbf{s}^{L^*(\mathbf{A})}, \mathbf{V}_{\mathbf{A}, 0} \right\rangle = \left\langle \mathbf{s}^L, \mathbf{A}_{n_2}^{L + n_1} \mathbf{A}_{n_1}^{L} \right\rangle.
    \]
\end{itemize}
Analogously for the remaining terms:
\begin{itemize}
    \item $\left\langle \mathbf{s}^{L^*(\mathbf{A})}, \mathbf{V}_{\mathbf{A}, 1} \right\rangle = \left\langle \mathbf{s}^L, \mathbf{B}_{n_2}^{L + n_1} \mathbf{A}_{n_1-1}^{L} \right\rangle$.
    \item $\left\langle \mathbf{s}^{L^*(\mathbf{B})}, \mathbf{V}_{\mathbf{B}, 0} \right\rangle = \left\langle \mathbf{s}^{L-1}, \mathbf{A}_{n_2}^{L + n_1}  \mathbf{B}_{n_1}^{L} \right\rangle$.
    \item $\left\langle \mathbf{s}^{L^*(\mathbf{B})}, \mathbf{V}_{\mathbf{B}, 1} \right\rangle = \left\langle \mathbf{s}^{L-1}, \mathbf{B}_{n_2}^{L + n_1}  \mathbf{B}_{n_1-1}^{L} \right\rangle$.
\end{itemize}
Thus
\begin{align*}
   \langle \mathbf{s}, \mathbf{P}_{L+r} \rangle= \langle \mathbf{s}^L, \mathbf{A}_r^L \rangle + \langle \mathbf{s}^{L-1}, \mathbf{B}_r^L \rangle &= \sum_{\tau_1 \in \{\mathbf{A},\mathbf{B}\}} \sum_{\epsilon_1=0}^{1} \left\langle \mathbf{s}^{L^*(\tau_1)}, \mathbf{V}_{\tau_1, \epsilon_1} \right\rangle\\
    &=\langle \mathbf{s}^L, \mathbf{A}_{n_2}^{L + n_1} \ \mathbf{A}_{n_1}^{L} \rangle + \langle \mathbf{s}^L, \mathbf{B}_{n_2}^{L + n_1}  \mathbf{A}_{n_1-1}^{L} \rangle \\
    &\quad + \langle \mathbf{s}^{L-1}, \mathbf{A}_{n_2}^{L + n_1}  \mathbf{B}_{n_1}^{L} \rangle + \langle \mathbf{s}^{L-1}, \mathbf{B}_{n_2}^{L + n_1} \mathbf{B}_{n_1-1}^{L} \rangle.
\end{align*}
Following this approach leads to a divide-and-conquer scheme that decomposes the original problem into more manageable computations, leveraging the recursive structure of the polynomials $A_r^L$ and $B_r^L$.\\%

Recall that the Fourier transform of $f$ on $\text{SU(2)}$ is the sequence of matrices $\{ \hat{f}(l) \}_{l \in \frac{1}{2} \mathbb{N}_0}$, where each matrix $\hat{f}(l) \in \mathbb{C}^{(2l+1) \times (2l+1)}$ is defined by
\begin{align*}
    \hat{f}(l)_{mn} := \int_{\text{SU(2)}} f(x) \; \overline{t^l_{nm}(x)} \, dx = \langle f, t^l_{nm} \rangle_{\ell^2(\text{SU(2)})},
\end{align*}
for $-l \leq m, n \leq l$, with $l - m, l - n \in \mathbb{Z}$.
So far, we have only computed the $(0, 0)$-entry ($m=n=0)$ of the matrix $\hat{f}(l)$, which corresponds to a single value in a matrix of size $2l+1$. To compute the remaining entries of the matrix, that is, in the case where $m \ne n \ne 0$, we consider the expression
\begin{align*}
    \hat{f}(l) \approx \sum_{k=0}^{N - 1} [\mathbf{s}]_k P^l_{nm}(\cos\theta_k) = \langle \mathbf{s}, \mathbf{P}^l_{nm} \rangle,
\end{align*}
where $P^l_{nm}$ is a Jacobi polynomial. Since these polynomials also satisfy a recurrence relation, an analogous approach to the one previously described can be applied.
\end{ex}
This procedure, which essentially constitutes an implementation of the Fast Fourier Transform on $\text{SU}(2)$, not only simplifies the computations but also enables future optimizations and applications in contexts where computational efficiency is critical.

\section{Order of Operations for the FT on SU(2)}\label{sec:order-FT}

We now analyze the computational complexity of computing the Fourier transform (FT) on the compact Lie group $\mathrm{SU}(2)$ using a direct (non-optimized) method. This corresponds to a brute-force evaluation of all Fourier coefficients using numerical quadrature over a discretized version.

\begin{thm}[Direct Computation Complexity]
Let $f : \mathrm{SU}(2) \to \mathbb{C}$ be a function bandlimited to degree $N$. If the Fourier coefficients $\hat{f}(l)_{mn}$ are approximated via numerical quadrature on a discretized grid of size $N^3$, the total computational complexity of the direct method is $\mathcal{O}(N^6)$.
\end{thm}

\begin{proof}
The Fourier transform consists of the collection of matrix coefficients $\hat{f}(l)_{mn}$ for $0 \leq l < N$. Since the dimension of the irreducible representation for a given $l$ is $2l+1$, the total number of coefficients to compute is
\begin{equation}
    \sum_{l=0}^{N-1} (2l+1)^2 = \sum_{l=0}^{N-1} (4l^2 + 4l + 1) = \mathcal{O}(N^3).
\end{equation}

Each coefficient is defined by the integral over the group, $\hat{f}(l)_{mn} = \int_{\mathrm{SU}(2)} f(g) \, \overline{t^l_{mn}(g)} \, dg$. Approximating this integral using a quadrature rule on a grid of size $N^3$ (derived from uniform discretizations of the Euler angles) yields the sum
\begin{equation}
    \hat{f}(l)_{mn} \approx \sum_{i=1}^{N^3} f(g_i) \cdot \overline{t^l_{mn}(g_i)}.
\end{equation}
Evaluating this sum requires $\mathcal{O}(N^3)$ operations per coefficient (assuming basis function evaluations are $\mathcal{O}(1)$).

Therefore, the total complexity is the product of the number of coefficients and the cost per coefficient
\begin{equation*}
    \underbrace{\mathcal{O}(N^3)}_{\text{num. coeffs}} \times \underbrace{\mathcal{O}(N^3)}_{\text{cost per coeff}} = \mathcal{O}(N^6). \qedhere
\end{equation*}
\end{proof}

This complexity makes the direct method impractical for large bandlimits $N$, and motivates the development of fast algorithms that the group structure and its representations to reduce the computational burden.

\section{Order of Operations for the FFT on SU(2)}\label{sec:order-fftsu2}

Throughout this section, we count one floating-point multiplication or addition as one arithmetic operation. We also analyze the computational complexity of the FFT on SU(2). The strategy is to decompose the algorithm into two main stages: a 2D FFT over $(\phi, \psi)$ and a weighted sum involving Legendre polynomials over $\theta$.

\begin{prop}[First Stage Complexity]\label{P_1}
Let $f$ be a function bandlimited to degree $N$, sampled on a uniform grid of size $N \times N \times N$ over the Euler angles $(\phi, \theta, \psi)$. The projection of the function onto the torus (integration over $\phi$ and $\psi$) via a 2D FFT can be computed in $\mathcal{O}(N^3 \log N)$ operations.
\end{prop}

\begin{proof}
Recall that the discretized expression for the Fourier coefficients is given by
\begin{equation*}
    \hat{f}(l)_{mn} \approx \sum_{k=0}^{N-1} P^l_{nm}(\cos \theta_k) \sin \theta_k \left( \sum_{j_1=0}^{N-1} \sum_{j_2=0}^{N-1} f(\phi_{j_1}, \theta_k, \psi_{j_2}) e^{in\phi_{j_1}} e^{im\psi_{j_2}} \right).
\end{equation*}

The term inside the parenthesis corresponds to a two-dimensional Discrete Fourier Transform (DFT) with respect to the variables $\phi$ and $\psi$. We can define the intermediate function $f_2(\theta_k; m, n)$ as the result of this transform
\begin{equation}
    f_2(\theta_k; m, n) := \sum_{j_1=0}^{N-1} \sum_{j_2=0}^{N-1} f(\phi_{j_1}, \theta_k, \psi_{j_2}) e^{in\phi_{j_1}} e^{im\psi_{j_2}}.
\end{equation}

Since the grid size is $N \times N \times N$, for any fixed $\theta_k$, the computation of $f_2$ is equivalent to a 2D FFT on an $N \times N$ grid. As established in the preliminaries, the cost of a single 2D FFT is $\mathcal{O}(N^2 \log N)$.\\%

The algorithm performs this operation independently for each of the $N$ discrete values of $\theta_k$. Therefore, the total cost for this stage is
\begin{align*}
    \text{Cost}_{\text{Stage 1}} &= (\text{Number of } \theta_k \text{ values}) \times (\text{Cost of 2D FFT}) \\
    &= N \times \mathcal{O}(N^2 \log N) \\
    &= \boldsymbol{\mathcal{O}(N^3 \log N)}. \qedhere
\end{align*}
\end{proof}

The problem is now reduced to computing
\begin{align}\label{EQ_OP_1}
    \hat{f}(l)_{mn}\approx\sum_{k=0}^{N-1}[s]_k P^l_{nm}(\cos \theta_k) = \langle \mathbf{s}, \mathbf{P}^l_{nm} \rangle,
\end{align}

where $[s]_k = f_2(\theta_k) \sin(\theta_k)$ is an arbitrary input vector $\mathbf{s}$ with its $k^{\text{th}}$ component equal to $[s]_k$. This must be performed for all relevant values of $l, m, n$.\\%

The remaining problem is the computation of the $k$-sum, which constitutes a Discrete Legendre Transform (or more precisely, a transform associated with the polynomials $P^l_{nm}$).

\begin{prop}[Second Stage Complexity]\label{P_2}
The computation of the Fourier coefficients $\hat{f}(l)_{mn}\approx  \langle \mathbf{s}, \mathbf{P}^l_{nm} \rangle$ via the recursive expansion of the Legendre polynomials of depth $k$ (where the shift $r$ is decomposed as $r = \sum_{i=1}^k n_i$) entails a total computational complexity of
\begin{equation*}
    \boldsymbol{\mathcal{O}(2^k k N^4)}.
\end{equation*}
\end{prop}

\begin{proof}
The problem is reduced to computing the inner product between an input vector $\mathbf{s}$ (where $[s]_k = f_2(\theta_k) \sin \theta_k$) and the polynomial vector $\mathbf{P}^l_{nm}$
\begin{equation*} 
    \hat{f}(l)_{mn} \approx \langle \mathbf{s}, \mathbf{P}^l_{nm} \rangle.
\end{equation*}

Using the recurrence relation for a shift of $r$ steps, decomposed into $k$ sub-shifts $n_1, \dots, n_k$, the inner product expands into the sum
\begin{equation*}
    \langle \mathbf{s}, \mathbf{P}_{L+r} \rangle = \sum_{\tau_1 \in \{\mathbf{A},\mathbf{B}\}} \sum_{\epsilon_1=0}^{1} \cdots \sum_{\epsilon_{k-1}=0}^{1} \left\langle \mathbf{s}^{L^*(\tau_1)}, \mathbf{V}_{\tau_1, \epsilon_1, \ldots, \epsilon_{k-1}} \right\rangle.
\end{equation*}

We analyze the cost in three steps\\%

\textbf{1. Number of Terms:}
For each $\tau_1 \in \{A, B\}$ and each binary vector $\vec{\epsilon} \in \{0,1\}^{k-1}$, there are distinct terms. The total number of terms in the summation is
\begin{equation*}
    2 \cdot 2^{k-1} = 2^k.
\end{equation*}

\textbf{2. Cost per Term:}
Each term involves computing the vector $\mathbf{V}_{\tau_1, \vec{\epsilon}}$ via Hadamard products
\begin{equation*}
    \mathbf{V}_{\tau_1, \vec{\epsilon}} =  \mathbf{\tau_1}_{n_1 - \epsilon_1}^{L} \odot \left( \bigodot_{i=2}^{k} \mathbf{\tau_i(\epsilon_{i-1})}_{n_i}^{L_{i-1}} \right) .
\end{equation*}
Then we are multiplying $k$ Hadamard vectors (one for each $i = 1, 2, \dots, k$)

\begin{itemize}
    \item $\mathbf{\tau_1}_{n_1 - \epsilon_1}^{L}$,
    \item \dots
    \item $\mathbf{\tau_k(\epsilon_{k-1})}_{n_k}^{L_{k-1}}$.
\end{itemize}

The Hadamard product of two vectors of length $L$ costs
\[
\mathcal{O}(L),
\]
because it requires $N$ scalar multiplications. Additionally, there are $k - 1$ Hadamard products in total, resulting in
\[
(k-1) \cdot \mathcal{O}(N) = \mathcal{O}(kN).
\]

Now consider the expression
\[
\mathbf{s}^{L^*(\tau_1)}.
\]
The input vector $\mathbf{s}^{L^*(\tau_1)}$ used in each inner product is defined as
\[
 \mathbf{s}^{L^*(\tau_1)} = 
 \begin{cases} 
 \mathbf{s} \cdot \mathbf{P}_L & \text{if } \tau_1 = \mathbf{A}, \\
 \mathbf{s} \cdot \mathbf{P}_{L-1} & \text{if } \tau_1 = \mathbf{B}.
 \end{cases}
\]
It is assumed that $\mathbf{s}^L$ and $\mathbf{s}^{L-1}$ are precomputed values before entering the recurrence (when we start advancing $r$ steps), therefore this projection costs 
\[
\mathcal{O}(N).
\]
Subsequently, the total cost of computing (assuming precomputed projected inputs).  Thus, the cost per term is
\[
\left\langle \mathbf{s}^{L^*(\tau_1)}, \mathbf{V}_{\tau_1, \epsilon_1, \ldots, \epsilon_{k-1}} \right\rangle,
\]
is
\[
   \mathcal{O}(kN + N) = \mathcal{O}(kN).
\]

\textbf{3. Total Complexity:}
The cost for a single coefficient $\hat{f}(l)_{mn}$ combines the number of terms and the cost per term
\begin{equation*}
    \text{Cost}_{\text{coeff}} = \underbrace{2^k}_{\text{terms}} \cdot \underbrace{\mathcal{O}(kN)}_{\text{cost/term}} = \mathcal{O}(2^k k N).
\end{equation*}
Since there are $\mathcal{O}(N^3)$ total coefficients to compute (summing over all $l, m, n$), the total complexity is
\begin{equation*}
    \mathcal{O}(N^3) \cdot \mathcal{O}(2^k k N) = \boldsymbol{\mathcal{O}(2^k k N^4)}. \qedhere
\end{equation*}
\end{proof}

With the complexity bounds for both the projection phase (Stage 1) and the recursive transformation (Stage 2) established, we can now synthesize these results to determine the overall computational cost of the algorithm. Since the complexity of the second stage depends on the recursion depth $k$, the following theorem aggregates these components to identify the optimal parameter $k$ that minimizes the computational burden.

\begin{thm}[Optimal Algorithm Complexity]\label{TEO_FIN}
The total computational complexity of the proposed algorithm for computing the Fourier transform on $\mathrm{SU}(2)$ is minimized by choosing a recursion depth of $k=1$, yielding a total cost of:
\begin{equation*}
    \boldsymbol{\mathcal{O}(N^4)}.
\end{equation*}
\end{thm}

\begin{proof}
The total computational operations $T(k)$ is the sum of the costs derived in the two previous propositions (the 2D FFT projection and the Recursive Legendre Transform). Combining these results, we obtain
\begin{align*}
    T(k) &= \text{Cost}_{\text{Stage 1}} + \text{Cost}_{\text{Stage 2}} \\
         &= \mathcal{O}(N^3 \log N) + \mathcal{O}(2^k k N^4).
\end{align*}
Expressing this with positive constants $C_1, C_2$, the cost function is
\begin{equation*}
    T(k) = C_1 N^3 \log N + C_2 2^k k N^4.
\end{equation*}
For large bandlimits $N$, the behavior is dominated by the second term, specifically the factor $h(k) = k 2^k$. We analyze the monotonicity of $h(k)$ for real $k > 0$ using its logarithmic derivative
\begin{equation*}
    \frac{d}{dk} \ln h(k) = \frac{d}{dk} (k \ln 2 + \ln k) = \ln 2 + \frac{1}{k}.
\end{equation*}
Since $\ln 2 > 0$ and $1/k > 0$ for all $k > 0$, the derivative is strictly positive. This proves that $h(k)$ is a strictly increasing function on the domain.\\%

Consequently, to minimize the computational burden, we must select the minimal valid integer for the recursion depth, which is
\begin{equation*}
    k_{\min} = 1.
\end{equation*}
Substituting $k=1$ into the total complexity expression, the dominant term becomes $2^1 \cdot 1 \cdot N^4$, and the total cost simplifies to
\begin{equation*}
    T(1) = \mathcal{O}(N^3 \log N) + \mathcal{O}(2 N^4) = \boldsymbol{\mathcal{O}(N^4)}. \qedhere
\end{equation*}
\end{proof}

It is also instructive to analyze the asymptotic behavior if the recursion depth is not kept constant but is instead allowed to grow logarithmically with respect to $N$. As shown below, this strategy fails to improve the overall complexity.

\begin{cor}
Attempting to reduce the exponential factor $2^{k}$ to a polylogarithmic factor by choosing $k$ dependent on $N$ such that $2^{k} \approx \log N$ (which implies $k \approx \log_2\log N$) leads to a less efficient algorithm. With this choice, the dominant term becomes
\begin{equation*}
    2^{k} k N^{4} \approx (\log N)(\log_2\log N)N^{4} = \boldsymbol{\mathcal{O}(N^{4}\log N\log\log N)}.
\end{equation*}
This complexity is asymptotically strictly greater than the $\mathcal{O}(N^4)$ bound achieved with constant $k=1$.
\end{cor}

\section{Computational Gap Between the Direct FT and the FFT on SU(2)}\label{sec:ftvsfft}

To emphasize the drastic improvement in computational efficiency offered by the fast Fourier transform (FFT) on SU(2), we compare its complexity with that of the direct method (FT). Both methods aim to compute all Fourier coefficients $\hat{f}(l)_{mn}$ up to a given bandlimit $N$.\\

The computational complexity of each method is given by
\begin{itemize}
    \item \textbf{Direct method (FT):} $\mathcal{O}(N^6)$
    \item \textbf{Fast algorithm (FFT-based):} $\mathcal{O}(N^3 \log N) +\mathcal{O}( 2^k k N^4)$.
\end{itemize}

For a comparative analysis, we must identify the dominant term in the FFT-based algorithm's complexity. As $N \to \infty$, the $N^4$ term grows asymptotically faster than $N^3 \log N$. Therefore, assuming $k$ is a non-zero constant, the complexity is dominated by the $\mathcal{O}(N^4)$ term. 

To quantify the performance gap, we calculate the approximate number of operations for various bandlimits, where $N$ is a power of two. The results are presented in Table \ref{tab:comparison}.

\begin{table}[H]
\centering
\caption{Comparison between the number of operations required by the direct FT on SU(2) ($\boldsymbol{O(N^6)}$) and the proposed FFT-based method for various values of $k$.}
\label{tab:comparison}
\begin{tabularx}{\textwidth}{|c|X|X|X|}
\hline
\rowcolor{red!15}
\textbf{Bandlimit $N$} & \textbf{FT Operations ($\boldsymbol{O(N^6)}$)} & \textbf{FFT,  ($\boldsymbol{O(N^4)}$), $k=1$} & \textbf{FFT, $k=\log_2 \log_2 N$ } \\
\hline

        \rowcolor{gray!10} 
$2^{10} = 1024$ & $\approx 1.15 \times 10^{18}$ & $\approx 1.07 \times 10^{12}$ & $\approx 3.55 \times 10^{13}$ \\
\rowcolor{gray!10}
$2^{12} = 4096$ & $\approx 4.72 \times 10^{21}$ & $\approx 2.82 \times 10^{14}$ & $\approx 1.21 \times 10^{16}$ \\
$2^{14} = 16384$ & $\approx 1.94 \times 10^{25}$ & $\approx 7.44 \times 10^{16}$ & $\approx 3.96 \times 10^{18}$ \\
\rowcolor{gray!10}
$2^{16} = 65536$ & $\approx 7.92 \times 10^{28}$ & $\approx 1.96 \times 10^{19}$ & $\approx 1.25 \times 10^{21}$ \\
\hline
\end{tabularx}
\end{table}

To better illustrate the practical implications of these differences, consider a processor speed of one gigaflop per second ($10^9$ operations/sec).  
For a practical bandlimit $N = 2^{10} = 1024$

\begin{itemize}
    \item The \textbf{direct FT} would require
    $$
    \frac{(2^{10})^6 \ \text{ops}}{10^9 \ \text{ops/sec}} 
    = \frac{1.15 \times 10^{18} \ \text{ops}}{10^9 \ \text{ops/sec}} 
    \approx 1.15 \times 10^{9} \ \text{seconds} \approx \textbf{36.5 years}.
    $$
    \item The \textbf{FFT-based method, $k=1$} would require
    $$
    \frac{(2^{10})^4 \ \text{ops}}{10^9 \ \text{ops/sec}} 
    = \frac{1.07 \times 10^{12} \ \text{ops}}{10^9 \ \text{ops/sec}} 
    \approx 1.07 \times 10^{3} \ \text{seconds} \approx \textbf{18 minutes}.
    $$
    \item The \textbf{FFT-based method, $k=\log_2 \log_2 N$} would require
    $$
    \frac{(2^{10})^4 \cdot (10) \cdot (3.32) \ \text{ops}}{10^9 \ \text{ops/sec}} 
    = \frac{3.55 \times 10^{13} \ \text{ops}}{10^9 \ \text{ops/sec}} 
    \approx 3.55 \times 10^{4} \ \text{seconds} \approx \textbf{9.86 hours}.
    $$
\end{itemize}

This corrected comparison shows that, for moderate $N$, the fast algorithm with a small recursion depth ($k=1$) can reduce the computation time from decades to minutes. Even with $k=\log_2 \log_2 N$, the runtime drops from decades to just a few hours for $
N=1024$. These estimates ignore memory access costs and cache effects, which affect both methods similarly and do not change the asymptotic comparison.\\



To visualize the magnitude of the computational improvement, Figure \ref{fig:ftvsfft} presents a comparison of the theoretical floating-point operations (FLOPs) required for both methods. The comparison is plotted on a logarithmic scale to highlight the difference in polynomial complexity classes.

\begin{figure}[!ht]
    \centering
    \begin{tikzpicture}
        \begin{loglogaxis}[
            width=0.85\textwidth,
            height=7cm,
            xlabel={Bandwidth $N$},
            ylabel={Theoretical Operations (FLOPs)},
            grid=both,
            minor grid style={gray!25},
            major grid style={gray!50},
            legend pos=north west,
            domain=4:64,
            samples=25,
            title={Complexity Analysis: Direct vs. FFT Variants},
            cycle list name=color list,
            legend cell align={left},
            legend style={font=\small}
        ]
        
        \addplot+[line width=1.2pt, mark=square, red] {x^6};
        \addlegendentry{Direct FT $\mathcal{O}(N^6)$}
        
        \addplot+[line width=1.2pt, mark=o, blue] {x^4};
        \addlegendentry{Base Complexity $\mathcal{O}(N^4)$}

        \addplot+[line width=1.2pt, mark=triangle, green!60!black, dashed] 
            {x^4 * (ln(x)/ln(2)) * (ln(ln(x)/ln(2))/ln(2))};
        \addlegendentry{Detailed $\mathcal{O}(N^4 \log N \log \log N)$}
        
        \end{loglogaxis}
    \end{tikzpicture}
    \caption{\textbf{Theoretical Complexity Comparison.} The red line shows the prohibitive scaling of the direct computation ($\mathcal{O}(N^6)$). The blue line represents the asymptotic $\mathcal{O}(N^4)$ complexity. The green dashed line accounts for the logarithmic factors in the recursive steps, following $\mathcal{O}(N^4 \log N \log \log N)$, which remains significantly more efficient than the direct method for large $N$.}\label{fig:ftvsfft}
\end{figure}
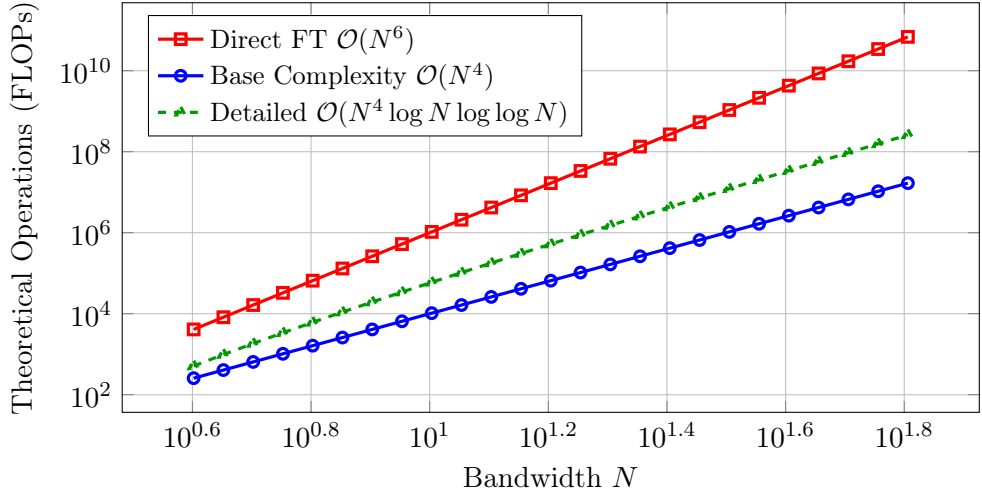

The graphical results confirm the substantial advantage of the proposed algorithm. In a log-log plot, the slope of the curve corresponds to the exponent of the complexity class. The direct method exhibits a steep slope ($\approx 6$), rendering it computationally intractable for large values of $N$ typically required in high-resolution quantum modeling. \\%

In contrast, our FFT implementation follows a significantly gentler slope ($\approx 4$). This reduction of two orders of magnitude in the exponent implies that for a bandwidth of $N=64$, the proposed method is theoretically orders of magnitude more efficient than the direct approach. This efficiency is crucial for applications in \textit{Mathematical Modelling} and opens new avenues for research in quantum computing. Since the classical FFT on $\mathrm{SU}(2)$ shares fundamental structural parallels with the Quantum Fourier Transform (QFFT), our algorithm serves as a vital structural prototype. Understanding these explicit classical recursions is a prerequisite for designing efficient quantum circuits and for developing advanced QFFT algorithms required in quantum simulation and quantum information processing.




\bibliographystyle{plain}

\bibliography{bib-Delgado-Lp-2016}

\end{document}